\documentclass[12pt,fleqn]{article}

\usepackage{natbib}
\usepackage{vmargin,setspace}
\usepackage{tabu}
\usepackage{amssymb,amsmath}

\usepackage{txfonts}

\usepackage[colorlinks=true, pdfstartview=FitV, citecolor=blue!50!black, linkcolor=red!30!black]{hyperref}

\usepackage{graphicx,graphics,tikz}
\usepackage{microtype}

\newcounter{llst}
\newenvironment{abet}{\begin{list}{\rm (\alph{llst})}{\usecounter{llst}
\setlength{\itemindent}{0em} \setlength{\leftmargin}{3em}
\setlength{\labelwidth}{2em} \setlength{\labelsep}{1em}}}{\end{list}}
\newenvironment{numm}{\begin{list}{\rm (\roman{llst})}{\usecounter{llst}
\setlength{\itemindent}{0em} \setlength{\leftmargin}{3.5em}
\setlength{\labelwidth}{2.5em} \setlength{\labelsep}{1em}}}{\end{list}}

\newtheorem{theorem}{Theorem}[section]

\newtheorem{corollary}[theorem]{Corollary}

\newtheorem{definition}[theorem]{Definition}
\newtheorem{expl}[theorem]{Example}

\newtheorem{proposition}[theorem]{Proposition}
\newtheorem{rmrk}[theorem]{Remark}

\newenvironment{proof}[1][Proof]{\noindent \textbf{#1.} }{\hfill
\rule{0.5em}{0.5em}}

\newenvironment{example}{\begin{expl} \rm}{\hfill $\blacklozenge$ \end{expl}}{}
{}

\setpapersize{A4}
\setmarginsrb{84pt}{64pt}{84pt}{72pt}{12pt}{10pt}{12pt}{30pt}

\begin{document}

\author{Owen Sims\thanks{Queen's University Management School, Riddel Hall, 185 Stranmillis Road, Belfast BT9~5EE, UK. Email: \href{mailto:osims01@qub.ac.uk}{osims01@qub.ac.uk}} \and Robert P. Gilles\thanks{Queen's University Management School, Riddel Hall, 185 Stranmillis Road, Belfast BT9 5EE, UK. Email: \href{mailto:r.gilles@qub.ac.uk}{r.gilles@qub.ac.uk}} }

\title{\textbf{Critical Nodes in Directed Networks}}

\date{January 2014}

\maketitle

\begin{abstract}
\singlespace\noindent
Critical nodes---or ``middlemen''---have an essential place in both social and economic networks when considering the flow of information and trade. This paper extends the concept of critical nodes to directed networks. We identify strong and weak middlemen.

Node contestability is introduced as a form of competition in networks; a duality between uncontested intermediaries and middlemen is established. The brokerage power of middlemen is formally expressed and a general algorithm is constructed to measure the brokerage power of each node from the networks adjacency matrix. Augmentations of the brokerage power measure are discussed to encapsulate relevant centrality measures.

We use these concepts to identify and measure middlemen in two empirical socio-economic networks, the elite marriage network of Renaissance Florence and Krackhardt's advice network.
\end{abstract}


\thispagestyle{empty}

\newpage

\setcounter{page}{1} \pagenumbering{arabic}

\section{The critical role of middlemen in social networks}

Recent research has shown that social networks are recognised as important descriptors of real-life social processes and their analysis as very insightful \citep{Watts2004,Jackson2008,Newman2010}. Node centrality aims to identify the most influential nodes in the network.

This paper makes a contribution to the analysis of nodes that are critical for the flow of information and trade in a network. Such critical nodes---or ``middlemen''---have been solely considered for undirected networks. Here, we consider these nodes in the general context of directed networks. We introduce a middleman as a node that can block the information flow from at least one node to another. If one applies this definition to undirected networks, one arrives at the standard notion of a middleman as a singleton node cut set. In the context of directed networks this is no longer the case: Its removal does not necessarily break up the whole network; it just compromises the information flow for at least one pair of nodes.

In directed networks, the essential brokerage position of middlemen allows them to be highly extractive to both directly and indirectly connected nodes. On the other hand, their involvement lowers transaction costs in terms of search and the formation of beneficial links \citep{Burt1992}. Examples include financial intermediaries which lower the cost of credit intermediation; Google which lowers the cost of Internet search and matching through self-selection; and eBay lowering the cost of economic trade.

Naturally, the existence of middlemen in a directed network is closely related to the competitive environment represented by that network. \citet{GillesDiam2013} show in a very simple setting that middlemen have the potential to be highly exploitive given a lack of competition combined with a pessimistic outlook regarding their potential contestability in a network. The main conclusion from this research is a non-trivial extension of how economists and social scientists perceive the architecture and dynamics of exchange systems, how the presence of a middleman can knit a system together, and, as a consequence of their position, can act as rent-extracting monopolists with excessive bargaining power \citep[Chapter 11]{EasleyKleinberg2010}.

We extend these notions to a more general definition of contestability in directed networks. In a directed network, a node is contested if alternative pathways (``walks'') are available to establish exchange between pairs of nodes in the network. Our main result shows that there is a formal duality between the existence of middlemen and network contestability. In particular, an intermediary node is a middleman if and only if it is uncontested.

One would expect that betweenness centrality measures would capture the influence that such middlemen exert. However, we show that this is not the case. Instead, we propose a new \emph{middleman power measure} that exhibits the desired properties. We devise an algorithm that applies the adjacency matrix of any directed network to identify and measure the middlemen in that network. We apply this algorithm to study two very well known (historical) directed networks from the literature. These applications allow us to make an in-depth comparison with well-established centrality measures such as betweenness centrality and Bonacich centrality \citep{Bonacich1987}.

In the Florentine marriage network of the early renaissance \citep{Padgett1993,Padgett1994}, we conclude that our middleman power measure clearly ranks the more powerful middlemen higher than the less powerful, confirming with historical analysis. In \citet{Krackhardt1987}'s advice network, again our measure confirms Krackhardt's original assessment of the most influential node. Furthermore, we show that, despite the importance of middlemen in these networks, this positional feature is not properly and fully identified by conventional centrality measures.

\paragraph{Relationship with the existing literature.}

Middlemen have been investigated in both sociology and economics during the past three decades. The specific interest in economics was initiated by \citet{KalaiMiddlemen1978} who investigated the payoffs to middlemen in an intermediated trade environment. These insights were extended by \citet{RubinsteinWolinsky1987} who applied them to a model of market search. \citet{JacksonWolinsky1996} and \citet{GillesChakrabarti2006} further analysed middlemen in the context of economic interaction.

Many of the general findings in the economics literature have been independently verified in the social sciences. In social network analysis, middlemen are seen as social actors that bridge structural holes in the social fabric. Due to their position these middlemen have access to more diverse information; are able to broker information and ideas between agents and affiliations; will have a tendency to be more entrepreneurial; are able to exploit their intermediating opportunities; and have better ideas which are then evaluated by others within a society or organisation \citep{Burt1992, Burt2004,Burt2005,Burt2010}. Moreover, by virtue of their weak ties in the network, middlemen are able to exploit unique opportunities that are not afforded to other, more socially embedded, nodes in the network \citep{Granovetter1973, Granovetter2005}.

Despite the wide acknowledgement within social network analysis of the significance of middlemen, centrality measures do not necessarily identify these critical nodes as being important even though their removal may deteriorate the functionality of the network as a whole. As the notion of centrality came to the fore, \citet[p.~219]{Freeman1979} argued that central nodes were those ``in the thick of things''. To exemplify this, he used an undirected star network consisting of five nodes. The middle node, at the centre of the star, has three advantages over the other nodes: It has more ties; it can reach all the others more quickly; and it controls the flow between the others. Freeman developed three simple measures of node centrality based on the prevalent features of the centre node: degree, closeness and betweenness.

More complex centrality measures have been proposed. \citet{Bonacich1972} proposed the use of the largest eigenvalue of the adjacency matrix as the basis for the measurement of network centrality. This measure expresses the idea that a node is more central if it is connected to nodes that are themselves central. Other measures, including Katz centrality \citep{Katz1953} take advantage of the same mechanism as Bonacich centrality by analysing the number of all nodes that a given node can be connected through a path, while contributions from linking distant nodes are progressively penalised. A node has a larger centrality if it has neighbours with a high Katz centrality. PageRank \citep{BrinPage1998} is a variant of Katz centrality, which recently gained application to the World Wide Web.

The most common centrality measures do not provide a proper account of network power and, usually, overlook middlemen. Two exceptions are noted: First, the $\beta$-measure \citep{BrinkGilles1994, BrinkGilles2000} is founded on an axiomatic approach to measuring centrality. This measure is one of a few that measures dominance, and therefore power, in directed networks. Applications are mainly focussed on how nodes are dominated by others in hierarchical organisations. Second, \citet{Bozzo2013} provide both a vulnerability and power measure based on the number of nodes that can be victimised by so-called ``executioners''. The power of a node set is comparable to how integral the set is to the robustness of an undirected network.

\paragraph{Outline of the paper.}

This paper is structured as follows. Section 2 provides the required preliminaries for our analysis and introduces two strong and weak middlemen in directed networks. Section 3 considers the notion of network contestability, which is shown to be the dual notion to that of middlemen. Section 4 provides a measure of middleman power, which assigns a quantitative expression to a node's brokerage power; non-middlemen always attain a zero network power score. The final section investigates two empirical case studies of social networks where middlemen have tended to have a degree of importance.

\section{Middlemen as critical nodes}

Our initial discussion focusses on the definition of a critical node in a directed network. We prefer to denote critical nodes as \emph{middlemen}. We introduce middlemen in relation to the connectivity in a network. We identify two related notions of middlemen, denoted as strong and weak middlemen depending on their network position.

\subsection{Network preliminaries}

A \emph{directed network} is a pair $(N,D)$ where $N = \{1,2, \ldots ,n\}$ is a finite set of \emph{nodes} and $D \subset \{ (i,j) \mid i,j \in N$ and $i \neq j \}$ is a set of \emph{arcs}, being directed relationships from one node to another. Note that $(i,i) \notin D$ for any $i \in N$, which is a technical requirement.\footnote{An \emph{undirected} network can be interpreted as a directed network $(N,D)$ such that all arcs are reciprocated: $(i,j) \in D$ if and only if $(j,i) \in D$.} We denote a directed network $(N,D)$ by $D$ unless $N$ is ambiguous. Also, we denote $(i,j) \in D$ by $ij$. We next introduce some auxiliary concepts in directed networks.

\paragraph{Walks:}

A \textit{walk} from $i$ to $j$ in a directed network $D$---or, simply, a $(i,j)$-\emph{walk}---is a set of connected nodes $W_{ij} (D) = \{ i_{1}, \ldots ,i_{m} \} \subset N$ with $m \geq 3$, $i_1 =i$, $i_m =j$, and $i_{k}i_{k+1} \in D$ for every $k=1, \ldots ,m-1$. Therefore, a walk is a sequence of adjacent nodes in the network. Walks might revisit nodes and, therefore, might contain loops.

In many cases there are multiple walks from $i$ to $j$ in a directed network $D$. If this is required, we denote $W_{ij}^{v}(D)$ as the $v^{\mbox{th}}$ distinct walk from $i$ to $j$ in $D$.

The class $\mathcal{W}_{ij}(D)= \{ W_{ij}^{1}(D), \ldots ,W_{ij}^{V}(D) \}$ consists of all distinct walks from $i$ to $j$ in $D$, where $V$ is the number of distinct walks. $\mathcal{W}_{ij}(D)= \varnothing$ denotes that there is no walk from $i$ to $j$ in $D$.

\paragraph{Connectedness:}

Nodes $i,j \in N$ are \textit{strongly connected} if $\mathcal{W}_{ij}(D) \neq \varnothing$ as well as $\mathcal{W}_{ji}(D) \neq \varnothing$. A directed network $D$ is \emph{strongly connected} if $\mathcal{W}_{ij} \neq \varnothing$ and $\mathcal{W}_{ji} \neq \varnothing$ for all nodes $i,j \in N$.

Node $i$ is \textit{weakly connected to} $j$ if $\mathcal{W}_{ij}(D) \neq \varnothing$. A directed network $D$ is weakly connected if for all $i,j \in N$ either $i$ is weakly connected to $j$, or $j$ is weakly connected to $i$, or both. Clearly, strongly connected networks are always weakly connected.

A subset $M \subset N$ is a \emph{weakly connected component} of $D$ if $(M, D_M)$ with $D_M = D \cap (M \times M)$ is weakly connected and and there is no $i \in N \setminus M$ with $(M+i, D_{M+i})$ being weakly connected.

\paragraph{Successors and predecessors:}

If $\mathcal{W}_{ij} (D) \neq \varnothing$ then $j$ is called the \textit{successor} of $i$ and $i$ is called the \textit{predecessor} of $j$ in $D$.

For $i \in N$ we define $s_{i}(D) = \{ j \in N \mid (i,j) \in D \}$ being all of the \textit{direct successors} of $i$ in $D$, and $p_{i}(D)=\{j \in N \mid (j,i) \in D\}$ being all the \textit{direct predecessors} of $i$ in $D$. The \emph{out-degree} of $i$ is given by $d_{i}^{+} = \# s_{i}(D)$ and its \emph{in-degree} by $d_{i}^{-}= \# p_{i}(D)$. Now, $i$'s overall degree is given by $d_{i} = d^+_i + d^-_i \equiv \# \{ s_{i} (D) \cup p_{i} (D) \}$.

We introduce $S_{i}(D)= \{j \in N \mid \mathcal{W}_{ij}(D) \neq \varnothing \}$ as $i$'s \textit{successor set}. Node $k$ is an \emph{indirect successor} of $i$ in $D$ if $k \in S_{i}(D) \setminus s_i (D)$. Furthermore, we introduce a modified successor set by $\overline{S}_{i}(D) = S_{i}(D) \cup \{i\}$.

Likewise, $P_{i}(D)=\{j \in N \mid \mathcal{W}_{ji}(D) \neq \varnothing \}$ denotes $i$'s \textit{predecessor set}. Again, node $k$ is an \emph{indirect predecessor} of $i$ in $D$ if $k \in P_{i}(D) \setminus p_{i}(D)$. Analogous to the above, we let $\overline{P}_{i}(D) = P_{i}(D) \cup \{i\}$.

We can also represent the position of a node in the network through a pair-based description. The \emph{coverage} of node $i$ in $D$ is the set of all pairs of nodes that $i$ intermediates, i.e., all pairs $(j.h) \in P_i (D) \times S_i (D)$. Similarly, the \emph{reach} of node $i$ in $D$ is defined by all pairs $(i,j)$ with $j \in S_i (D)$.

The node set $N$ can be partitioned into three disjoint subsets: sources, sinks, and intermediaries. Node $i$ is a \emph{source} if $d_{i}^{-} = 0$ and $d_{i}^{+} \geqslant 1$; $i$ is a \emph{sink} if $d_{i}^{-} \geqslant 1$ and $d_{i}^{+} = 0$; and $i$ in an \emph{intermediary} if $d_{i}^{-} \geqslant 1$ and $d_{i}^{+} \geqslant 1$.

Finally, we let $D-i$ represent the network obtained by deleting the node $i$ from the network $D$. This is equivalent to
\begin{equation}
D - i = \{ (j,h) \mid j,h \in N \setminus \{ i \} \mbox{ such that } (j,h) \in D \, \}.
\end{equation}

\subsection{Connectivity and middlemen}

We identify \textit{critical nodes} as those having the ability to broker information flows in a network. Following \citet{GillesChakrabarti2006}, a critical node in an undirected network can disrupt and manipulate the typical operations on a network by disconnecting two or more nodes---which is equivalent to the property that removing a critical node partitions a given network into multiple disconnected components.

Here, we extend this concept to directed networks. In this context our definition focusses on the disruption of connectivity of two nodes.
\begin{definition} \label{middleman}
Let $(N,D)$ be a directed network and let $i,j,h \in N$ be three distinct nodes.
\begin{abet}
\item The node $h$ is an \textbf{$(i,j)$-middleman} if $\mathcal{W}_{ij} (D) \neq \varnothing$ and
\begin{equation} \label{ijmiddleman}
h \in \cap \mathcal{W}_{ij}(D) = W_{ij}^{1}(D) \cap  \cdots  \cap W_{ij}^{V}(D) ,
\end{equation}
where $V \geqslant 1$ is the number of walks from $i$ to $j$. Here, $M_{ij}(D) \equiv \cap \mathcal{W}_{ij}(D)$ denotes the \textbf{$(i,j)$-middleman set} containing all $(i,j)$-middlemen.

\item The \textbf{middleman set} for $D$ is the set of all middlemen in $D$ given by
\begin{equation} \label{middlemanseteq}
M(D) = \bigcup_{i,j \in N \colon i \neq j} M_{ij}(D)
\end{equation}
Node $h$ is a \textbf{middleman} if $h \in M(D)$ and $h$ is a \textbf{non-middleman} if $h \notin M(D)$.
\end{abet}
\end{definition}
A middleman is an intermediary node that is a member of \emph{all} walks between at least two other nodes within a given network. Therefore, a middleman brokers the flow of all information between at least two nodes. Obviously, a middleman connects two or more agents that would not be connected otherwise. Indeed, by definition $M_{ij} (D) \neq \varnothing$ if and only if $\mathcal{W}_{ij} (D) \neq \varnothing$.

Converse to a middleman, a non-middleman is an intermediary that if removed from the network does not affect the connectivity of any two or more other nodes and, therefore, does not impose a negative externality on other nodes when removed.

The following properties can be deduced directly from the above definition.
\begin{proposition}\label{TheoremIntermediary}
Let $D$ be a directed network on the node set $N = \{ 1, \ldots ,n \}$.
\begin{numm}
\item Every middleman $i \in M(D)$ is an intermediary in $D$. 

\item If $n < 3$, there might exist intermediaries, but there are no middlemen.

\item For $n \geqslant 3$ and $i \in N$, $p_{i}(D) \subset \left[ p_{j}(D) \cup \{ j \} \, \right]$ for every direct successor $j \in s_{i}(D)$ implies that $i \notin M(D)$.

\item For all $i,j \in N$ with $j \in s_{i}(D)$ it holds that $M_{i,j}(D) = \varnothing$.
\\
This implies that the complete directed network has no middlemen.

\item Every middleman $i \in M(D)$ has a local clustering co-efficient of less than $1$.

\item If $D$ is undirected in the sense that $(i,j) \in D$ if and only if $(j,i) \in D$, then $M_{ij} (D) = M_{ji} (D)$ for all $i,j \in N$.
\end{numm}
\end{proposition}
From (vi), in an undirected network a node is a middleman if it rests on all walks from node $i$ to node $j$. This means that the removal of a middleman disrupts the any communication between nodes $j$ and $i$. Hence, a middleman indeed is a singleton cut set in an undirected network, as traditionally understood.

On the other hand, in a directed network, a middleman rests on all walks from node $i$ to node $j$, but does not have to rest on all walks from $j$ to $i$. This implies that given the removal of a middleman in a directed network the interaction from $i$ to $j$ will become disconnected, but the interaction from node $j$ to node $i$ may still be present. Indeed, even with the removal of a middleman in a directed network, nodes $i$ and $j$ can still remain weakly connected. This insight motivates a further refinement of the notion of a middleman in a directed network.
\begin{definition} \label{strongweakmiddlemen}
Let $D$ be a weakly connected directed network on node set $N = \{ 1, \ldots ,n \}$ with $n \geqslant 3$.
\begin{abet}
\item A middleman $h \in M(D)$ is a \textbf{strong middleman} in $D$ if the network $D - h$ contains two or more weakly connected components.

\item A middleman $h \in M(D)$ is a \textbf{weak middleman} in $D$ if network $D-h$ does not contain two or more weakly connected components.
\end{abet}
\end{definition}
Weak middlemen exist in both cyclic and acyclic directed networks. Example~\ref{identifyingmiddlemen} highlights the existence of both weak and strong middlemen in an acyclic network.

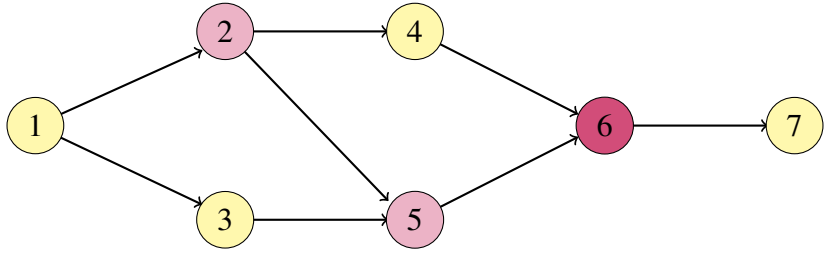
\begin{figure}[t]
\begin{center}
\begin{tikzpicture}[scale=0.5]
\draw[thick, ->] (0,2.5) -- (4.4,4.5);
\draw[thick, ->] (0,2.5) -- (4.4,0.4);
\draw[thick, ->] (5,5) -- (9.3,5);
\draw[thick, ->] (5,0) -- (9.3,0);
\draw[thick, ->] (5,5) -- (9.3,0.5);
\draw[thick, ->] (10,5) -- (14.3,2.8);
\draw[thick, ->] (10,0) -- (14.3,2.2);
\draw[thick, ->] (15,2.5) -- (19.3,2.5);

\draw (0,2.5) node[draw,circle,fill=yellow!40] {$1$};
\draw (5,5) node[draw,circle,fill=purple!30] {$2$};
\draw (5,0) node[draw,circle,fill=yellow!40] {$3$};
\draw (10,5) node[draw,circle,fill=yellow!40] {$4$};
\draw (10,0) node[draw,circle,fill=purple!30] {$5$};
\draw (15,2.5) node[draw,circle,fill=purple!70] {$6$};
\draw (20,2.5) node[draw,circle,fill=yellow!40] {$7$};

\end{tikzpicture}
\end{center}
\caption{Acyclic directed network, $D$, highlighting strong, weak, and non-middlemen.}
\label{weakmm}
\end{figure}

\begin{example} \label{identifyingmiddlemen}
Consider a directed network $D$ on the node set $N=\{1,2,3,4,5,6,7\}$, which is depicted in Figure~\ref{weakmm}. In the network, node $1$ is a source, node $7$ is a sink, and nodes 2,3,4,5, and 6 are intermediaries.
\\
We can determine that $M(D)=\{2,5,6\}$, where nodes 2 and 5 are weak middlemen and node 6 is a strong middleman. For an explanation first consider node 2. Node 2 lies on all walks from node 1 to node 4, therefore $2 \in M_{1,4}(D)$ and $\mathcal{W}_{1,4}(D-2) = \varnothing$. However, $D-2$ remains weakly connected meaning that node 2 must be a weak middleman. An analogous argument could be made for node 5 because $5 \in M_{3,6}(D) \cap M_{3,7}(D)$ and, therefore, $\mathcal{W}_{3,6}(D-5) = \mathcal{W}_{3,7}(D-5) = \varnothing$. Moreover, $D-5$ is weakly connected.
\\
Finally, $6 \in M_{1,7}(D) \cap M_{2,7}(D) \cap M_{3,7}(D) \cap M_{4,7}(D) \cap M_{5,7}(D)$. Indeed, node 6 is the sole broker of all interaction to node 7. In the network $D-6$ there emerge two weakly connected components: $\{ 1,2,3,4,5 \}$ and $\{ 7 \}$.
\\
All other intermediaries are non-middlemen. Indeed, even the removal of either node 3 or 4 does not affect the connectivity in the network.\footnote{It is worth noting that if all arcs were reciprocated in network $D$ from Figure~\ref{weakmm} to form an undirected network, nodes 2 and 5 would no longer be middlemen. However, node 6 would still be a strong middleman and would subsequently broker twice as many interactions, i.e. all connections to node 7 and all connections from node 7 to the rest of the network.}
\end{example}
Theorem~\ref{undirectedmiddlemen} below naturally leads from the insights resulting in Definition~\ref{strongweakmiddlemen} and Example~\ref{identifyingmiddlemen}.
\begin{theorem} \label{undirectedmiddlemen}
Every middleman in an undirected network is a strong middleman.
\end{theorem}
\begin{proof}
Without loss of generality we can restrict ourselves to connected undirected networks only. First note that every connected undirected network is a strongly connected directed network $D$ with $ij \in D$ if and only if $ji \in D$ and there exists at least one walk from $i$ to $j$ for all $i \neq j$.
\\
Therefore, consider a strongly connected directed network $(N,D)$ where $\# N =n \geqslant 3$. According to Definition~\ref{middleman}, a node $h \in N$ is a middleman if it rests on all walks between at least two other nodes, say $i$ and $j$. Since $W_{ij}(D) = W_{ji}(D)$, the property that $h \in \bigcap_{i,j \in N} \mathcal{W}_{ij}(D)$ implies that $h \in \bigcap_{i,j \in N} \mathcal{W}_{ji}(D)$.
\\
Thus, in $D-h$ all walks from $j$ to $i$ are disconnected. This in turn implies that $i$ and $j$ cannot be weakly connected and $D-h$ must contain at least two weakly connected components, separating $i$ and $j$ in different components. This implies that $h$ is actually a strong middleman in $D$.
\end{proof}

\medskip\noindent
It should be intuitive why Theorem~\ref{undirectedmiddlemen} holds. Weak middlemen exist in directed networks due to the distinction between weakly connected and strongly connected nodes. The distinction collapses in an equivalent undirected network as all nodes are effectively strongly connected. This implies the following.
\begin{corollary} \label{corundirectedmiddleman}
Weak middlemen only exist in directed networks.
\end{corollary}
The distinction between weak and strong middlemen is natural and enhances our understanding of the functionality of directed versus undirected networks. We enhance this understanding further in the following discussion that allows the measurement of middleman control, in which it is shown that weak middlemen can actually be more powerful than strong middlemen.

\section{Contestability in directed networks}

Next we examine the relationship between critical nodes and competition in networks. Based on the model of network competition in \citet{GillesDiam2013}, such competition rests on the ability of actors to circumvent intermediaries in their pursuit of value-generating interaction. Thus, competition in a network is understood as the ability to prevent an intermediary from becoming a middleman.

In directed networks, contestability is modelled as the ability of a group of nodes to service the coverage of an intermediary, given by the product of the nodes' predecessor and successor set. A node is contested by other nodes if this group of nodes can cover all connections facilitated by that node.
\begin{definition} \label{Contested}
Let $D$ be a directed network on node set $N=\{1, \ldots ,n\}$ where $n \geqslant 3$ and let $i \in N$.
\begin{abet}
\item Node $i$ is \textbf{contested} by node set $C_{i} \subset N$ if $i \notin C_i$ and it holds that
\begin{equation} \label{Group Contested}
P_{i}(D) \times S_{i}(D) \subseteq \bigcup_{j \in C_{i}} \left( \, \overline{P}_{j}(D-i) \times S_{j}(D-i) \, \right).
\end{equation}

\item Node $i$ is \textbf{directly contested} by $j \neq i$ if the singleton node set $\{ j \}$ contests $i$ in $D$.

\item Node $i$ is \textbf{uncontested} if there are no sets of nodes that contest $i$.
\end{abet}
\end{definition}
There may exist multiple node sets that contest $i$ in a network $D$. This justifies the introduction of a class $\mathcal{C}_i (D) \subset 2^N$ of such contesting node sets. Furthermore, a \emph{minimal} contesting node set is given by $C_{i}^{*}(D) \in \arg \min \{ \#C_{i} \mid C_{i} \in \mathcal{C}_{i} \}$.
\begin{corollary}
A node $i$ is directly contested by node $j$ in network $D$ if and only if\begin{equation} \label{Directly Contested}
P_{i}(D) \subseteq P_{j} (D - i) \cup \{j\} \quad \mbox{and} \quad S_{i} (D) \subseteq S_{j} (D - i) \cup \{j\}.
\end{equation}
\end{corollary}
The corollary states the explicit nature of contestation in a network in that a node can completely take over the functionality of the contested node. Thus, node $i$ is directly contested by node $j$ only when all of node $i$'s predecessor set can be connected to $i$'s successor set either through or from node $j$ when $i$ is removed from the network. The exact same intuition is used with respect to contestation by a group of nodes.
\begin{example} \label{Simple Contestability}
We consider a network to illustrate the notion of contestability. Consider directed network $D$ on node set $N = \{1,2,3,4,5,6,7\}$ shown in Figure~\ref{weakmm} on page \pageref{weakmm}. Table 1 below provides the predecessor and successor sets of all nodes in the network.

\begin{table}[h]
\begin{center}
\label{network1stats}
\begin{tabu}{ l l l }

\tabucline[2pt]{-}
Node & Predecessor Set                 & Successor Set                     \\ \hline
1    & $P_{1}(D)=\varnothing$          & $S_{1}(D)=\{2,3,4,5,6,7\}$        \\
2    & $P_{2}(D)=\{1\}$                & $S_{2}(D)=\{4,5,6,7\}$            \\
3    & $P_{3}(D)=\{1\}$                & $S_{3}(D)=\{5,6,7\}$              \\
4    & $P_{4}(D)=\{1,2\}$              & $S_{4}(D)=\{6,7\}$                \\
5    & $P_{5}(D)=\{1,2,3\}$            & $S_{5}(D)=\{6,7\}$                \\
6    & $P_{6}(D)=\{1,2,3,4,5\}$        & $S_{6}(D)=\{7\}$                  \\
7    & $P_{7}(D)=\{1,2,3,4,5,6\}$      & $S_{7}(D)=\varnothing$            \\ \hline
\end{tabu}\par
\caption{Predecessor and successor sets of nodes in Figure~\ref{weakmm}}
\end{center}
\end{table}

\noindent
Using this information we deduce that intermediaries 3 and 4 are contested, whereas intermediaries 2, 5, and 6 are uncontested.
\\
Here, node 3 is contested by node 2: The predecessor set of node 3 is given by $P_{3}(D) = \{1\} \equiv P_{2}(D)$ and the successor set of node 3 is given by $S_{3}(D) = \{5,6,7\} \subset S_{2}(D) = \{4,5,6,7\}$. It is also true that $P_{3}(D) \subseteq P_{2}(D-3) \cup \{2\} = \{ 1,2 \}$ and $S_{3}(D) \subseteq S_{2}(D-3) \cup \{2\} = \{ 2,4,5,6,7 \}$.
\\
This case introduces what can be denoted as \textit{asymmetric contestability}, meaning that although node $i$ contests node $j$, it may not be true that node $j$ contests node $i$. Here, node 2 directly contests by node 3, although node 2 is not directly contested by node 3. Only in rare cases will there exist \textit{symmetric contestability} where node $i$ contests node $j$ and node $j$ contests node $i$.
\end{example}
The next example highlights group contestability where a highly connected node asymmetrically contests two others, while these two nodes in turn contest the highly connected node.

\begin{example} \label{Group Contestability}
Consider directed network $D'$ on node set $N = \{1,2,3,4,5,6\}$, shown in Figure~\ref{Complex Contestability}, where $M(D') = \varnothing$. Here, node $4$ connects nodes $1$ and $2$ to nodes $5$ and $6$, and therefore directly contests node $2$ while not being directly contested by any other individual node. However, node $4$ is not a middleman and indeed the set $ C= \{ 2,3 \}$ contests node $4$.

\begin{figure}[h]
\begin{center}
\begin{tikzpicture}[scale=0.5]
\draw[thick, ->] (5,13) -- (9.5,5.6);
\draw[thick, ->] (5,13) -- (5,2.9);
\draw[thick, ->] (10,10) -- (0.8,10);
\draw[thick, ->] (10,10) -- (0.8,5.5);
\draw[thick, ->] (10,10) -- (5.6,2.6);
\draw[thick, ->] (10,5) -- (0.8,5);
\draw[thick, ->] (10,5) -- (0.8,9.7);
\draw[thick, ->] (5,2) -- (0.7,4.5);
\draw[thick, ->] (5,2) -- (0.5,9.3);

\draw (5,13) node[draw,circle,fill=yellow!40] {$1$};
\draw (10,10) node[draw,circle,fill=yellow!40] {$2$};
\draw (10,5) node[draw,circle,fill=yellow!40] {$3$};
\draw (5,2) node[draw,circle,fill=yellow!40] {$4$};
\draw (0,5) node[draw,circle,fill=yellow!40] {$5$};
\draw (0,10) node[draw,circle,fill=yellow!40] {$6$};

\end{tikzpicture} 
\caption{Network $D'$ where node set $C = \{ 2,3 \}$ contests node $4$.}
\label{Complex Contestability}
\end{center}
\end{figure}
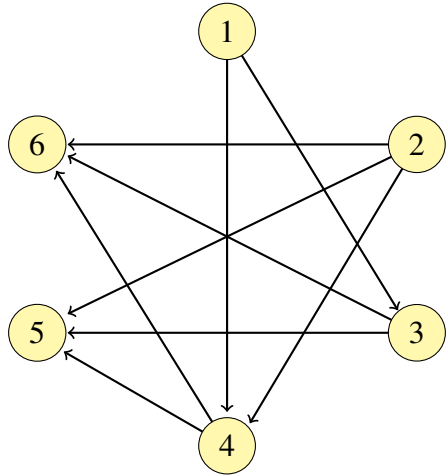

\noindent
Clearly, the coverage and reach of nodes $2$ and $3$ encapsulate the coverage of node $4$. Therefore, although nodes $2$ and $3$ do not contest node $4$ individually, the node set $C = \{ 2,3 \}$ contests node $4$. Indeed, the condition for group contestability holds:
\begin{equation}
P_{4}(D) \times S_{4}(D) \subset \left( \, \overline{P}_{2}(D-4) \times S_{2}(D-4) \, \right) \cup \left( \, \overline{P}_{3}(D-4) \times S_{3}(D-4) \, \right).
\end{equation}
If node $4$ is removed from the network its function can be fully replaced by the combination of nodes $2$ and $3$ and therefore all other nodes that were connected can still be connected in the same way.
\end{example}
Example~\ref{Group Contestability} highlights the requirement for $\overline{P}_{j}(D-i)$ used in the definition of contestability instead of the predecessor set, $P_{j}(D-i)$, only. Consider the network in Figure~\ref{Complex Contestability}. As noted, there exist no middlemen and all intermediaries are contested given the definition of contestability provided above. With a more restricted definition it would be seen that node $4$ would not be a contested intermediary and also not a middleman. The more elaborative version provided in Definition~\ref{Contested} adjusts for predecessors of a given node, $i$, that can connect to the successors of $i$, therefore fulfilling the same function and thus contesting $i$. 

Examples~\ref{Simple Contestability} and~\ref{Group Contestability} give an indication that if an intermediary is contested, it cannot be a middleman. For example, in Figure~\ref{weakmm} agent $3$ is a non-middleman because his function is directly contested by the presence of node $2$, and node $2$ is a middleman because its function is not contested by any other node in the network.

Our main result states that there is a duality between contestation and the presence of critical nodes.
\begin{theorem} \label{duality}
Consider any directed network $(N,D)$ with $n \geqslant 3$. Then:
\begin{numm}
\item Every middleman $i \in M(D)$ is uncontested.

\item If an intermediary $i \in N$ is uncontested, then it must be a middleman, i.e., $i \in M(D)$.
\end{numm}
\end{theorem}
\begin{proof}
Let $(N,D)$ be a directed network with $n \geqslant 3$.

\smallskip\noindent
\emph{Proof of (i):}
The condition for contestability stated in equation (\ref{Group Contested}) on page \pageref{Group Contested} states that a node $h \in N$ is contested in network $D$ if its coverage---determined by the pairs of nodes that it intermediates---is a subset of the coverage plus the reach of the nodes in $C_{h} \subset N \setminus \{ h \}$.
\\
Now consider an intermediary $i$ in $D$ that is contested by a set of agents, $C_{i} \subset N \setminus \{ i \}$. Since $i$ is contested, it must be true that all of $i$'s predecessors can be connected to all of $i$'s successors by a walk that does not include $i$. Therefore, $i$ cannot be a middleman. This implies the assertion that every middleman is uncontested.

\smallskip\noindent
\emph{Proof of (ii):}
Consider an intermediary $h \in N$ who is uncontested in the network $D$. Then $h$'s coverage is not a subset of the coverage of any set of nodes plus the respective reach of each of these nodes. This implies that $h$ itself has to rest on at least one walk that no other nodes in the network rest on when $h$ is removed from the network. Hence, in the network $D$ there exists at least one pair of nodes, say $i$ to $j$, with $h \in \cap \mathcal{W}_{ij} (D)$ and $\mathcal{W}_{ij}(D-h) = \varnothing$. This implies that $h$ is actually a middleman concerning the walks from $i$ to $j$.
\end{proof}

\medskip\noindent
From this assertion all middlemen are never contested; if a node is contested then all of its intermediation functions can be replaced by the coverage of other nodes. From this, it is contended that a middleman is an intermediary that has a unique function and is, in some way, more effective than non-middlemen with respect to their connectivity and thus coverage in the network.

\paragraph{Relationship with market competition.}

Considering an economic network in which all nodes produce an homogeneous output, the notion of contestability is linked to that of market competition. Traditional market theory contends that one agent competes against another if they produce the same output and subsequently sell this output to the same set of consumers. If two or more agents operate in a given market with access to the same set of suppliers and the same consumers, then all consumers could technically be supplied by another agent if one of the agents were removed from the network, or failed due to some exogenous shock.

Bertrand competition suggests that firms producing an homogenous output will be prone to compete with each other with respect to the price attached to their respective outputs \citep{Edgeworth1889}. Much of the analysis in market competition in economics has derived from this simple concept. Firms operating in this way will have no market power; potentially a price above the long-term marginal cost would be unsustainable since a competitor will be able to undercut it. Our notion of contestability is related to Bertrand competition in the sense that, if a given node is contested, its coverage and functions are contested, meaning that a contested node has no power in the network.\footnote{Note that if all nodes provided a heterogeneous output then the notion of contestability and market competition breaks down. Furthermore, if we were to assume that all agents had highly differentiated information and ideas then individual nodes could technically not contest one another since they would all perform different functions and provide different insights into the network; indeed, all nodes will be unique.}

\section{Measuring Middleman power}
\label{networkpower}

A middleman occupies a critical position in a network since her removal disconnects at least two or more agents and, in the most extreme case, might have the ability to separate the network into multiple disconnected components. Therefore, it seems logical to ask how we can measure this power. After examining established measures, we propose a measure of middleman power based on the disconnections that emerge when middlemen are removed from the network.

\subsection{Betweenness centrality and middlemen}

We first examine whether betweenness centrality could be a tool to assess middleman power. \emph{Betweenness centrality} was proposed independently by \citet{Anthonisse1971} for edges and rephrased by \citet{Freeman1977} for nodes in undirected networks. \citet{White1994} proposed an extension to directed networks.

This measure seems specifically relevant since it explicitly considers the role of a node in connecting other nodes in the network. It may be expected that the betweenness centrality score of a node provides an indication of what nodes are middlemen by having a greater betweenness centrality than non-middlemen in the network.

To define the betweenness centrality measure, let $\pi(hj)$ be the number of shortest walks---or \emph{geodesics}---from node $h$ to node $j$. Furthermore, let $\pi_{i}(kj)$ be the number of geodesics that pass through node $i$. Betweenness centrality is now defined as
\begin{equation} \label{betweennesscentrality}
BC_{i}(D) = \sum_{h,j \neq i \colon \pi(hj) \neq 0} \frac{\pi_{i}(hj)}{\pi(hj)}
\end{equation}
Equation (\ref{betweennesscentrality}) indicates that a middleman $i$ between nodes $k$ and $j$ would always have a high betweenness centrality since by Definition~\ref{middleman} a middleman is on all walks between these two nodes. However, this formulation only considers geodesics. As the next example illustrates, the betweenness centrality of non-middlemen may even exceed the betweenness centrality of middlemen.
\begin{example} \label{ex:BC}
Consider the acyclic directed network $D''$ depicted in Figure~\ref{mmnmm}.

\begin{figure}[h]
\begin{center}
\begin{tikzpicture}[scale=0.5]
\draw[thick, ->] (0,10) -- (4.1,10);
\draw[thick, ->] (0,0) -- (4.1,0);
\draw[thick, ->] (0,5) -- (4.1,5);

\draw[thick, ->] (5,10) -- (9.2,8);
\draw[thick, ->] (5,10) -- (9.2,3);

\draw[thick, ->] (5,5) -- (9.2,7.3);
\draw[thick, ->] (5,5) -- (9.1,2.5);

\draw[thick, ->] (5,0) -- (9.5,6.8);
\draw[thick, ->] (5,0) -- (9.2,2);

\draw[thick, ->] (10,7.5) -- (14,7.5);
\draw[thick, ->] (10,2.5) -- (14,2.5);

\draw[thick, ->] (10,7.5) -- (14,2.7);
\draw[thick, ->] (10,2.5) -- (14,7.4);

\draw (0,0) node[draw,circle,fill=yellow!40] {$3$};
\draw (0,5) node[draw,circle,fill=yellow!40] {$2$};
\draw (0,10) node[draw,circle,fill=yellow!40] {$1$};
\draw (5,0) node[draw,circle,fill=yellow!40] {$6$};
\draw (5,5) node[draw,circle,fill=yellow!40] {$5$};
\draw (5,10) node[draw,circle,fill=yellow!40] {$4$};
\draw (10,7.5) node[draw,circle,fill=yellow!40] {$7$};
\draw (10,2.5) node[draw,circle,fill=yellow!40] {$8$};
\draw (15,7.5) node[draw,circle,fill=yellow!40] {$9$};
\draw (15,2.5) node[draw,circle,fill=yellow!40] {$10$};

\end{tikzpicture}
\end{center}
\caption{The directed network $D''$ considered in Example~\ref{ex:BC}.}
\label{mmnmm}
\end{figure}
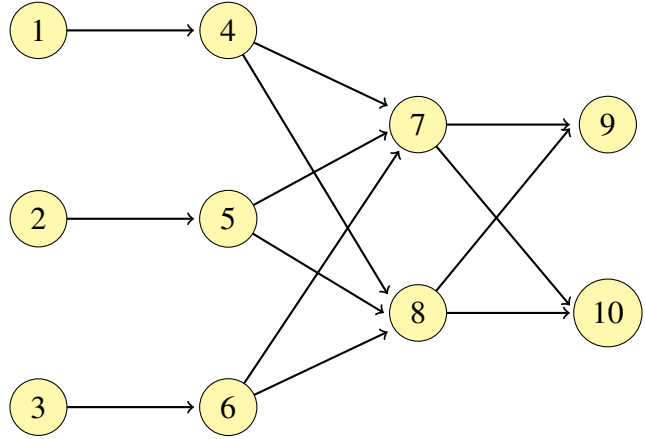

\noindent
Here $M(D'') = \{ 4, 5, 6 \}$, and nodes $7$ and $8$ are contested intermediaries. All middlemen have the same non-normalised betweenness centrality measures due to their equivalent positions: $BC_{4}(D'')=BC_{5}(D'')=BC_{6}(D'')=4$. However, both contested intermediaries have larger non-normalised betweenness centrality scores: $BC_{7}(D'')=BC_{8}(D'')=6$.
\\
In the underlying undirected network, $U''$, where all arcs in $D''$ are reciprocated, the non-middlemen still have a higher betweenness centrality than middlemen: $BC_{4}(U'')=BC_{5}(U'')=BC_{6}(U'')=16.4$ and $BC_{7}(U'')=BC_{8}(U'')=25$. Table 2 provides a comparison of common centrality measures all of which indicate that nodes $7$ and $8$ are more central, therefore underrating the nodes with powerful middleman properties.
\end{example}
This deficiency concerning the identification and measure of middlemen, highlighted in Figure~\ref{mmnmm} and Table 2, extends to other less common centrality measures. Indeed, no measure for undirected networks specifically identifies and highlights middlemen, instead the measures tend to over-inflate the power and importance of contested intermediaries in networks.

The reason for the high betweenness centrality of non-middlemen, for example, is because of the underlying assumptions of the measure. First, only geodesic paths are counted between two given nodes, and the second is that these geodesic paths are considered to have an equal weight. The combination of these assumptions implies that the betweenness centrality measure does not necessarily measure the ``power'' of a node in negotiating between two others.

\begin{table}[t]
\begin{center}
\label{network2stats}
\begin{tabu}{ c c c c c c c }

\tabucline[2pt]{-}
Node & Degree & PageRank& Betweenness& Closeness& Bonacich& $\beta$-Measure\\ \hline
1    & $1$    & $0.112$ & $0.000$    & $0.360$  & $0.328$ & $0.333$\\
2    & $1$    & $0.112$ & $0.000$    & $0.360$  & $0.328$ & $0.333$\\
3    & $1$    & $0.112$ & $0.000$    & $0.360$  & $0.328$ & $0.333$\\
4    & $3$    & $0.330$ & $0.456$    & $0.474$  & $1.047$ & $1.400$\\
5    & $3$    & $0.330$ & $0.456$    & $0.474$  & $1.047$ & $1.400$\\
6    & $3$    & $0.330$ & $0.456$    & $0.474$  & $1.047$ & $1.400$\\
7    & $4$    & $0.480$ & $0.694$    & $0.692$  & $1.565$ & $2.000$\\
8    & $4$    & $0.480$ & $0.694$    & $0.692$  & $1.565$ & $2.000$\\
9    & $2$    & $0.297$ & $0.011$    & $0.474$  & $0.863$ & $0.400$\\
10   & $2$    & $0.297$ & $0.011$    & $0.474$  & $0.863$ & $0.400$\\ \hline
\end{tabu}\par
\caption{Centrality results for the undirected network $U''$}
\end{center}
\end{table}

\subsection{A middleman power measure}

We provide a mechanism to identify middlemen and quantify their brokerage power in a meaningful way from basic information regarding the topology of the network. The proposed measure counts the disconnections that emerge when a given node is removed from the network.

Let $D$ be a directed network on node set $N = \{1, \ldots ,n\}$ and let $i \in N$ be an arbitrary node. The \emph{brokerage} of node $i$ is quantified as
\begin{equation} \label{brokerage}
b_{i}(D) = \sum_{j \in N} \# S_{j}(D) - \sum_{j \neq i} \# S_{j}(D - i) - \#S_{i}(D) - \#P_{i}(D) .
\end{equation}
Note that the successor set of a node contains all other nodes that can be reached by a directed walk from the node. The first part of Equation \ref{brokerage}, $\sum_{j \in N} \# S_{j}(D)$, counts the total number of successors of all $n$ nodes in the network. Hence, it provides an indication of the total connectivity of the network as a whole.

Given the intuition of the first part of the equation, it can be implied that the second part, $\sum_{j \neq i} \# S_{j}(D - i)$, refers to the total connectivity of the network when node $i$ is removed. We remark that $\sum_{j \in N} \# S_{j}(D) > \sum_{j \neq i} \# S_{j}(D - i)$ if $d_{i}(D) \geqslant 1$, and $\sum_{j \in N} \# S_{j}(D) = \sum_{j \neq i} \# S_{j}(D - i)$ if $d_{i} = 0$.

Thus, $\sum_{j \in N} \# S_{j}(D) - \sum_{j \neq i} \# S_{j}(D - i)$ expresses the \emph{connectivity differential} between the full network $D$ and the subnetwork $D-i$. The connectivity differential captures two features: (1) The direct connectivity of node $i$ in terms of its successor and predecessor set. Indeed, the larger the predecessor and successor sets of node $i$ the larger the differential will be regardless of whether $i$ is a middleman or not. (2) The lost connectivity to other nodes not including $i$.

When considering the impact of a middleman, we are only interested in the lost connectivity (2) and, therefore, we compensate the connectivity differential with the connectivity of $i$ in the network. Specifically, the predecessor set and successor set of node $i$ is removed from the connectivity differential, thus adjusting for (1) and resulting in Equation~\ref{brokerage} above. The brokerage of a node therefore counts the number of third-party disconnections that occur due to the removal of a node; or in other words, counts the number of $(i,j)$-middleman sets that a node is a member of.

If $b_{i}(D) = 0$ then the removal of $i$ from the network makes no change to the network's connectivity---when compensating for the connectivity of $i$. Hence, all nodes that are connected by a directed walk in $D$ can still be connected in $D-i$. In essence, compensating for their connection to $i$ in $D$, the number of successors of all $j$ nodes is the same in $D-i$ as in $D$.

On the other hand, if $b_{i}(D) \geqslant 1$, there exists at least one pair of connected nodes that are now not connected in $D-i$, and $i$ must be a middleman.

\medskip\noindent
We normalise the brokerage of a node by calculating the total number of potential opportunities for brokerage in the network. Brokerage---and, therefore, middleman positions---can only emerge if a pair of nodes are a minimum distance of two or more away from each other. Intuitively, by calculating the indirect successors of all nodes in the network, the total number of brokerage opportunities can be derived.

The set of indirect successors of $i$ in $D$ is given by $S_{i}(D) \setminus s_{i}(D)$. Therefore, the number indirect successors for node $i$ is given by $\# S_{i}(D) - \# s_{i}(D)$. Given this, the maximal potential brokerage in $D$ is computed as
\begin{equation} \label{normalisation}
B'(D) = \sum_{i \in N} \left[ \# S_{i}(D) - \# s_{i}(D) \right] .
\end{equation}
By normalising a node's brokerage score, the \emph{middleman power} of a node can be defined.
\begin{definition} \label{middlemanpower}
Let $D$ be a directed network on node set $N$. The \textbf{middleman power} of node $i \in N$ is given as
\begin{equation} \label{mmpowerindex}
\nu_{i}(D) = \frac{b_{i}(D)}{B(D)} ,
\end{equation}
where $B(D) = \max \{ B'(D) , 1\}$.
\end{definition}
Empty and complete networks have to be assigned an artificial normalisation factor of unity as they would otherwise have no brokerage opportunities. 

A middleman has a network power of 1 if it brokers all potential opportunities in the network. This includes nodes at the center of star networks; however, even if some of the leaf nodes form a connection between each other, the middleman power of the centre node remains 1.

The next example explicitly computes the middleman power for an undirected star and a directed cycle.

\begin{example} \label{starcycle}
Figure~\ref{star} shows an undirected star, $D^{\star}$, and a directed cycle, $D^{\circ}$, where $n=6$ in both networks. For an arbitrary undirected star network, $D^{\star}$, where $n \geqslant 3$, $b_{i}(D^{\star}) = (n-1)(n-2)$ for the centre node, and $b_{j}(D^{\star})=0$ for all other nodes. The potential brokerage for an undirected star is computed as $B(D^{\star}) = (n-1)(n-2)$. Therefore, the middleman power of the centre node is $\nu_{i}(D^{\star}) = \frac{(n-1)(n-2)}{(n-1)(n-2)} = 1$.

\begin{figure}[h]
\begin{center}
\begin{tikzpicture}[scale=0.5]
\draw[thick] (4,3.5) -- (1.5,0);
\draw[thick] (4,3.5) -- (6.5,0);
\draw[thick] (4,3.5) -- (0,5);
\draw[thick] (4,3.5) -- (8,5);
\draw[thick] (4,3.5) -- (4,8);

\draw (1.5,0) node[draw,circle,fill=yellow!40] {$4$};
\draw (6.5,0) node[draw,circle,fill=yellow!40] {$3$};
\draw (4,3.5) node[draw,circle,fill=yellow!40] {$6$};
\draw (0,5) node[draw,circle,fill=yellow!40] {$2$};
\draw (8,5) node[draw,circle,fill=yellow!40] {$5$};
\draw (4,8) node[draw,circle,fill=yellow!40] {$1$};

\draw[thick, ->] (15,8) -- (12.4,6.2);
\draw[thick, ->] (12,5.5) -- (12,3.25);
\draw[thick, ->] (12.5,2.5) -- (14.3,0.45);
\draw[thick, ->] (15,0) -- (17.5,1.9);
\draw[thick, ->] (18,2.5) -- (18,4.7);
\draw[thick, ->] (18,5.5) -- (15.7,7.6);

\draw (12,2.5) node[draw,circle,fill=yellow!40] {$4$};
\draw (18,2.5) node[draw,circle,fill=yellow!40] {$5$};
\draw (15,0) node[draw,circle,fill=yellow!40] {$6$};
\draw (12,5.5) node[draw,circle,fill=yellow!40] {$2$};
\draw (18,5.5) node[draw,circle,fill=yellow!40] {$3$};
\draw (15,8) node[draw,circle,fill=yellow!40] {$1$};

\end{tikzpicture}
\end{center}
\caption{Undirected network ($D^{\star}$) and directed cycle ($D^{\circ}$).}
\label{star}
\end{figure}
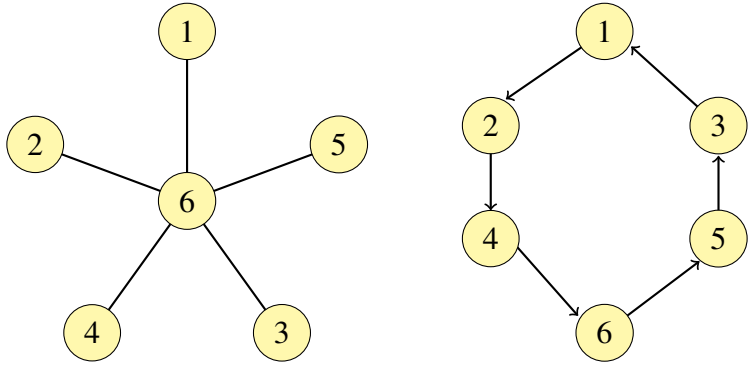

\noindent
Next, consider a directed cycle $D^{\circ}$ for arbitrary $n\geqslant 3$. Each node has an in-degree of 1 and an out-degree of 1, implying that all nodes are middlemen. We now derive that $b_{1}(D^{\circ}) = \ldots = b_{n}(D^{\circ}) = \frac{(n-1)(n-2)}{2}$. The number of potential brokerage opportunities is $B (D^{\circ}) = n(n-2)$, implying $\nu_i = \frac{n-1}{2n}$ for all $i \in N$ where $n \geqslant 3$. In particular, for the depicted case of $n=6 \colon \nu_i = \frac{5}{12}$ $\forall i \in N$.
\end{example}
We derive several properties for the middleman power measure stated in (\ref{mmpowerindex}).
\begin{theorem} \label{middlemanpowert}
Let $D$ be a network on node set $N=\{1, \ldots ,n\}$.
\begin{abet}
\item[(i)] For any contested intermediary, $\nu_{i}(D) = 0$.
\item[(ii)] For any network $0 \leqslant \nu_{i} \leqslant 1$ for all $i \in N$.
\end{abet}
\end{theorem}
\begin{proof}
\ \ \\
\emph{Proof of (i):}
Theorem~\ref{duality} asserts a duality between a contested intermediary and a non-middleman. Definition~\ref{middleman} implies that an intermediary $k$ is not a middleman if there is no pair $i,j \in N$ with $i \neq j$ such that $k$ lies on all walks from $i$ to $j$ in $D$. Hence, $k \notin \cap \mathcal{W}_{i,j}(D)$ for all $i,j \in N$ with $i \neq j$.
\\
Since $k$ is an intermediary it holds that $\# P_{k}(D) > 0$, $\# S_{k}(D) > 0$ as well as $\sum_{i \in N} \# S_{i}(D) > \sum_{i \in N-k} \# S_{i}(D-k)$ since the nodes in $D-k$ can obviously not connect to $k$. Also, since the connectivity of the network $D$ is not affected by the removal of node $k$, it holds that the removal of $k$ only affects the connectivity with $k$ itself. Hence, $\sum_{i \in N} \# S_{i}(D) - \sum_{i \in N-k} \# S_{i}(D-k) = \# S_{k}(D) + \# P_{k}(D)$. Therefore, $b_k (D) =0$, implying that $\nu_{k} = 0$.
\\[1.5ex]
\emph{Proof of (ii):}
We omit a mathematical proof of (ii), due to its tedious nature. Instead, we provide an intuitive, more descriptive reasoning. 
\\
A middleman cannot take advantage of more than all brokerage opportunities present in a network; therefore $B'(D) \geqslant b_{i}(D)$, implying $\nu_{i}(D) \leqslant 1$.\footnote{Only in a network where a middleman rests on all geodesic walks of length two, for example an undirected star, it holds that $B'(D) = b_{i}(D)$.}
\\
Furthermore, neither $b_{i}(D) < 0$ nor $B(D) < 0$. The minimum brokerage of some node $k$ is in an empty network where $\# S_{k}(D) = \# P_{k}(D) = 0$. In that case, $\sum_{i \in N} \# S_{i}(D) = \sum_{i \neq k} \# S_{i}(D-k)$ since $k$ has no connectivity in the network. Therefore, $\nu_i (D) \geqslant 0$ for any node $i \in N$.
\end{proof}

\medskip\noindent
For large directed networks the calculation of the middleman power measure of a node can be tedious. The appendix to this paper contains a two-step algorithm based on the adjacency matrix of a directed network to identify middlemen and compute the middleman power measure of each node in an arbitrary directed network.

\subsubsection*{Distance-based brokerage}

We extend our discussion of the criticality of nodes with a measure that combines middleman power with node proximity. Consider a directed network $D$ on $N$ and $i,j,h \in N$ with $h \in M_{ij}(D)$. The brokerage power of middleman $h$ could be less effective due to the geodesic distance from $i$ to $j$.

Consider an amended brokerage score to capture this effect given by
\begin{equation}
\Delta_{ij}(h) = \frac{1}{\delta_{ih}} \cdot \frac{1}{\delta_{hj}},
\end{equation}
where $\delta_{ij}$ is the geodesic distance from $i$ to $j$ in $D$. Here, nodes closer to $h$ provide a greater brokerage power to node $h$ than those at larger distances. Indeed, $h$ receives maximal brokerage power if $i \in p_{h}(D)$ and $j \in s_{h}(D)$.

Now a \emph{distance based middleman power measure} for $h \in M(D)$ can be introduced as
\begin{equation}
\nu^{\ast}_h (D) = \sum_{i,j \in N \colon h \in M_{ij} (D)} \Delta_{ij} (h) .
\end{equation}
Reassessing the directed network in Figure~\ref{weakmm}, the distance-based middleman power measure provides a convergence of the middleman power scores for the nodes: $\nu^{\ast}_{6}(D) = 3 \frac{1}{3}$; $\nu^{\ast}_{5}(D) = 1 \frac{1}{2}$; $\nu^{\ast}_{2}(D) = 1$; $\nu^{\ast}_{1}(D) = \nu^{\ast}_{3}(D) = \nu^{\ast}_{4}(D) = \nu^{\ast}_{7}(D) = 0$.

It may be particularly beneficial to use this modified measure to assess costly trade in a network or the diffusion of information that can degrade as it is being passed through a network. This assumption of information degradation and even complete truncation over a certain distance has been widely used in literature regarding social networks \citep{JacksonRogers2005, Jackson2008}.

\medskip \noindent
We note that the middleman power measure should not be considered as a replacement for other centrality measures---it is itself not just a measure of centrality---rather it identifies a certain type of node in a network and measures brokerage. Instead, the measure can be complimented with other measures of centrality. The above example of distance-based brokerage is one of many augmentations that makes use of closeness centrality.

\section{An application to two empirical networks}

In this section we apply our two middleman power measures to two well-known social networks. From the assessment of these networks we provide a discussion regarding the potential of middlemen in these networks. The results of middleman power are compared with other measures of centrality. This is done in terms of reference; we refrain from correlating the results of these measures because we showed above that middleman power measures different aspects of a node than other measures.

\subsection{The elite Florentine marriage network}

The investigation of the marriage network of Renaissance Florence has been extensively used to assess many centrality measures \citep{Newman2003betweenness, Borgatti2005}. We apply our analysis of critical nodes to this network as well.

\begin{figure}[t]
\centering
\includegraphics[scale=0.37]{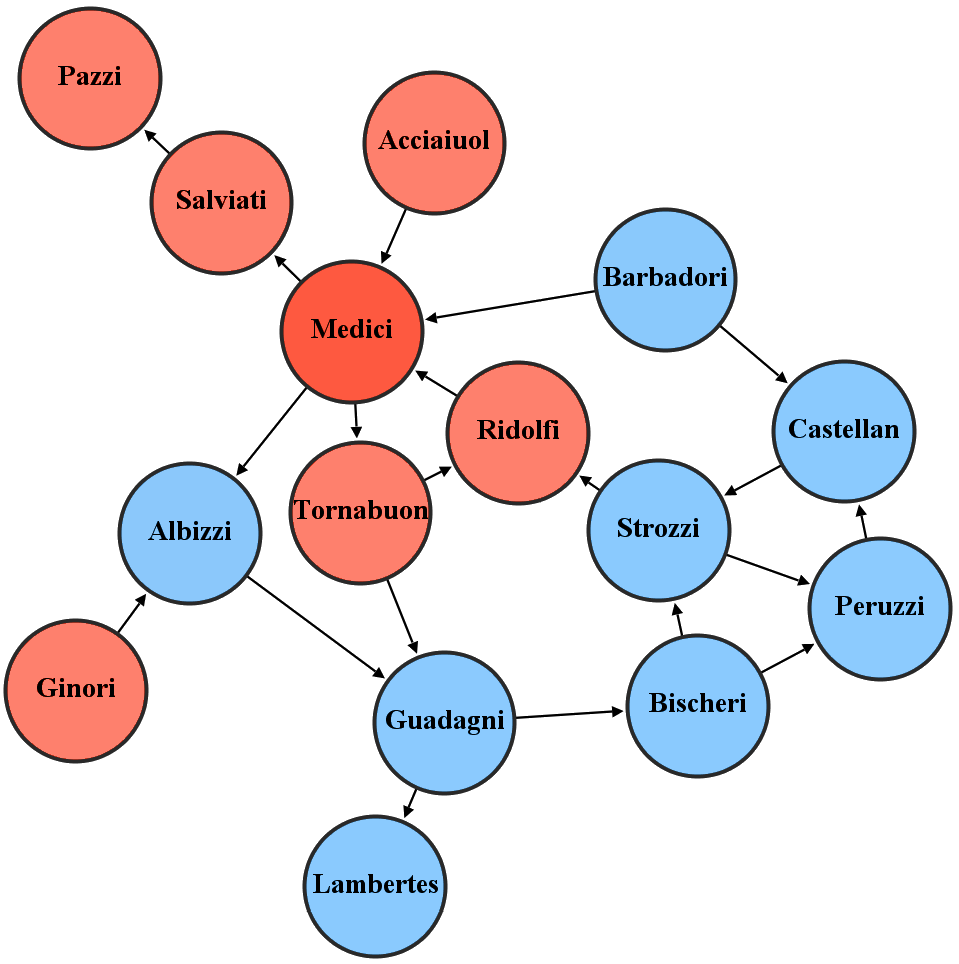}
\caption{Directed network of Florentine marriages, circa 1434}
\label{Fig:FlorentineFamilies}
\end{figure}

Although it has been shown previously that the Medici family have the highest degree, closeness, and betweenness centralities, their importance and power in the socio-economic network---as explained by \citet{deRoover1963}, \citet{Padgett1994} and \citet{Goldthwaite2009}---is derived from their ability to access diverse sources of information between separate factions of the Florentine elite and subsequently broker relationships between families. Powerful brokerage opportunities emerged due to the inherent ``network disjunctures within the elite'' \citep[p.1259]{Padgett1993}. In particular, \citet{Padgett1993} show that Cosimo de'Medici was able to gain access to, and control, the flow of diverse information between opposing political factions and also between families in the same faction. Thus, a middleman position allowed the Medici family to attain power within Florentine society; especially, the Medici's ability to act as broker between a large number of families crossing opposing political factions.

Unlike more recent renditions of the Florentine marriage network, we represent this network as a directed graph: An arc from family $i$ to family $j$ represents a female from family $i$ married into family $j$. The intuition for this representation is that families were strategically married in order to build trust, inherit wealth, property, and business, as well as influence the political and economic decisions of other families. Information will have flowed through these relationships. Moreover, the nodes are coloured depending on the factions that the families were affiliated: red nodes are families affiliated with the Medician faction, and blue nodes are families affiliated with the opposing Oligarch faction \footnote{The data was initially gathered by \citet{Kent1978} and a blockmodel network was constructed by \citet{Padgett1993}. The network provided in Figure~\ref{Fig:FlorentineFamilies} is derived from these studies. Both provide an extremely rich analysis of the elite in Florence at this time.}.

We analyse the resulting directed network in Figure~\ref{Fig:FlorentineFamilies} using in-degree ($d^-$) and out-degree ($d^+$), \citet{Bonacich1987} eigenvector centrality with a $\beta$ of $0.4$ ($E$), betweenness centrality ($BC$), raw middleman power ($\nu$) and distance-based middleman power ($\nu^*$) presented in Table 3. In this table, (*) denotes a weak middleman and (**) denotes a strong middleman.

\begin{table}
\begin{center}
\label{tabFlorence}
\begin{tabu}{ X  X  X  X  X  l  }
\tabucline[2pt]{-}
Family            & $d^-_i \ (d^+_i)$ & $E_i$    & $BC_i$   & $\nu_i$  & $\nu_{i}^{*}$\\ \hline
Medici(**)        & 3 (3)             & 1.756    & 0.376    & 0.586    & 0.223\\
Salviati(**)      & 1 (1)             & 0.275    & 0.066    & 0.103    & 0.064\\
Pazzi             & 1 (0)             & 0.000    & 0.000    & 0.000    & 0.000\\
Acciaiuol         & 0 (1)             & 0.977    & 0.000    & 0.000    & 0.000\\
Barbadori         & 0 (2)             & 1.549    & 0.000    & 0.000    & 0.000\\
Ridolfi(*)        & 2 (1)             & 0.977    & 0.264    & 0.414    & 0.092\\
Tornabuon         & 1 (2)             & 1.359    & 0.102    & 0.000    & 0.033\\
Albizzi(**)       & 2 (1)             & 0.694    & 0.149    & 0.095    & 0.033\\
Ginori            & 0 (1)             & 0.552    & 0.000    & 0.000    & 0.000\\
Castellan(*)      & 2 (1)             & 0.743    & 0.069    & 0.086    & 0.044\\
Strozzi(*)        & 2 (2)             & 1.169    & 0.245    & 0.371    & 0.112\\
Bischeri(*)       & 1 (2)             & 1.246    & 0.203    & 0.319    & 0.090\\
Peruzzi(*)        & 2 (1)             & 0.571    & 0.049    & 0.077    & 0.022\\
Guadagni(**)      & 2 (2)             & 1.048    & 0.269    & 0.422    & 0.130\\
Lambertes         & 1 (0)             & 0.000    & 0.000    & 0.000    & 0.000\\ \hline
\end{tabu}
\caption{Measuring the importance of Florentine families (1429 -- 1434)}
\end{center}
\end{table}

Assessing the marriage network with our middleman power measure highlights the Medici family as a strong middleman, along with the prominent Albizzi, Salviati, and Guadagni families. Clearly, the Medici family broker substantially more relationships than the next rival family, the Guadagni, and therefore control the flow of information and influence between more families. Moreover, it could be that the Medici would be able to control interactions in the form of marriages and economic trade between factions that may integrate them and subsequently lessen the Medici's critical position in the network.

The betweenness centrality measure ranks the Medici family as the most central in the network. The measure favours middlemen thus typically ranking them higher; however there is an exception: although the Tornabuon family are not middlemen they have a higher betweenness centrality than the Salviati and Peruzzi families.

Unsurprisingly, the Bonacich centrality measure also ranks the Medici family highly, but also ranks many non-middlemen highly; specifically the Barbadori family. From the network topology alone there is no indication to suggest that the Barbadori family should have had a prominent role in the Florenine elite. Although the Medici rank highly with this measure, the relevance of an eigenvector centrality measure is questionable: there is no real reason to believe why the importance of a family would come from its degree and the degree of its neighbours alone.

There are three further points to note. First, if the marriage network were converted into its undirected form, the Medici retains its middleman position maintaining the highest middleman power. Second, the Medici would also retain the highest ranking if the direction of the relationships were inverted. Third, \citet{Padgett1993} and \citet{PadgettMcLean2006} noted that the oligarch faction was less cohesive and more disorganised than the Medician faction. From our analysis it is seen that there are more brokerage opportunities in the oligarchic faction highlighting that they were more disjunctured. We see that $\tfrac{3}{4}$ of the oligarchs are middlemen against $\tfrac{3}{7}$ of the families in the Medician faction.

An additional remark is that the Medici attained control in Florence because they not only brokered relationships between elite families in the network, but because they brokered relationships between two opposing factions. Therefore, they were not only able to gather and intercept information between the factions and use it to their advantage. This is the prerogative of middlemen in social networks; to not only intercept information between agents, but to manipulate and exploit it.

\begin{figure}[t]
\centering
\includegraphics[scale=0.45]{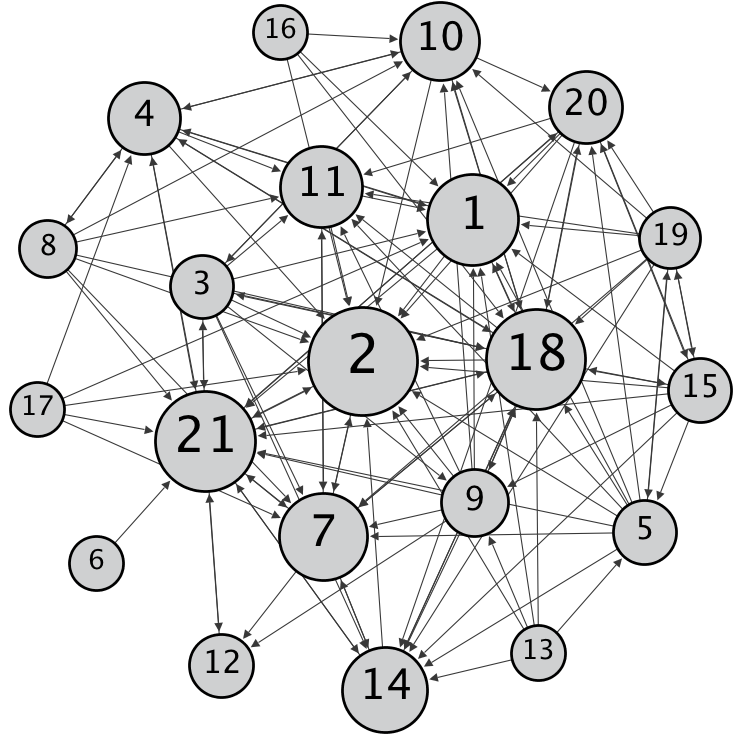}
\caption{Krackhardt's network of advice among managers.}
\label{krackhardtnetwork}
\end{figure}

\subsection{Middlemen in Krackhardt's advice network}

Next we consider an example of a well-known organisational advice network in which \citet{Krackhardt1987} investigated the relationships between managers in a given firm.\footnote{We use data in the ``LAS'' matrix from p.129 in the Krackhardt article as it seems to be the most objective measure.} He considered an organisation of about 100 employees and 21 managers. Krackhardt collected information from the managers about who sought advice from whom, depicted in Figure~\ref{krackhardtnetwork} on page \pageref{krackhardtnetwork}. An arc from $i$ to $j$ denotes that manager $i$ has sought advice from manager $j$; therefore, an arc from $j$ to $i$ denotes that manager $j$ has provided advice to manager $i$. In this depiction the size of the node reflects its in-degree.

\begin{table}[t]
\begin{center}
\label{tabkrackhardt}
\begin{tabu}{ X X X X X l }

\tabucline[2pt]{-}
Manager       & $d^-_i$ ($d^+_i$)& $E_i$    & $BC_i$   & $\nu_i$  & $\nu_{i}^{*}$     \\ \hline
1             & 12 (4)           & 0.068    & 0.035    & 0.000    & 0.000\\
2             & 18 (2)           & 0.306    & 0.011    & 0.000    & 0.000\\
3             & 3 (9)            & 1.271    & 0.018    & 0.000    & 0.000\\
4(*)          & 6 (7)            & 1.001    & 0.071    & 0.090    & 0.034\\
5             & 3 (10)           & 1.463    & 0.009    & 0.000    & 0.000\\
6             & 0 (1)            & 0.172    & 0.000    & 0.000    & 0.000\\
7             & 11 (6)           & 0.776    & 0.048    & 0.000    & 0.000\\
8             & 1 (7)            & 1.013    & 0.001    & 0.000    & 0.000\\
9             & 4 (9)            & 1.171    & 0.011    & 0.000    & 0.000\\
10            & 8 (5)            & 0.820    & 0.018    & 0.000    & 0.000\\
11            & 9 (3)            & 0.344    & 0.004    & 0.000    & 0.000\\
12            & 3 (1)            & 0.172    & 0.000    & 0.000    & 0.000\\
13            & 0 (6)            & 0.938    & 0.000    & 0.000    & 0.000\\
14            & 10 (4)           & 0.625    & 0.002    & 0.000    & 0.000\\
15(*)         & 3 (9)            & 1.265    & 0.092    & 0.161    & 0.064\\
16            & 0 (4)            & 0.580    & 0.000    & 0.000    & 0.000\\
17            & 0 (5)            & 0.673    & 0.000    & 0.000    & 0.000\\
18            & 15 (12)          & 1.745    & 0.231    & 0.000    & 0.000\\
19            & 2 (10)           & 1.493    & 0.002    & 0.000    & 0.000\\
20            & 6 (7)            & 1.028    & 0.028    & 0.000    & 0.000\\
21(**)        & 15 (8)           & 1.348    & 0.176    & 0.147    & 0.051\\ \hline
\end{tabu}
\caption{Influence, centrality, and middlemen in Krackhardt's advice network}
\end{center}
\end{table}

In Table 4, using the same centrality measures as before, we identify two weak middlemen, managers 4 and 15, and one strong middleman, manager 21.\footnote{As before, (*) indicates a weak middleman and (**) indicates a strong middleman.} Middlemen are important for this network for a number of intuitive reasons: First, a middleman can block ideas, advice, and information from being transmitted from one group of managers to another. Second, a middleman can manipulate the information transferred from one group of managers to another.

The weak middleman, Node 15, has the highest brokerage in the organisation controlling a total of 34 relationships. This is also reflected in that \citet{Krackhardt1987} highlighted manager 15 as an important agent in the organisational advice network. However, Node 15 does not have the highest betweenness or Bonacich centralities. Instead, Node 18 is the most prominent in terms of Bonacich and betweenness scores but is not a middleman in the network; instead this may be a function of the node 18's in- and out-degree.

Both Bonacich and betweenness centralities are also seen to be poor indicators for ranking middlemen: Node 15 is the most powerful middleman in the network, but has a Bonacich and betweenness centrality lower than Node 21. The Bonacich influence model does not consider the fact that middlemen are potentially able to exploit their position by using information from others and blocking the transmission of certain information and ideas.


\singlespacing
\bibliographystyle{econometrica}
\bibliography{OSDB}

\newpage

\appendix
\section*{Appendix: Algorithm for calculating middleman power}

We can calculate the network power of a node from the adjacency matrix using a multi-stage algorithm. The first algorithm used is provided below using the MATLAB programming language.\footnote{The devised algorithm is general enough to be configured for any programming language.}

\begin{scriptsize}
\begin{verbatim}
Language: MATLAB
Input: Adjacency Matrix (D)
Output: Raw (NP) and normalised (np) network power of individual nodes
Function: [NP,np] = middleman_power(D)

1   D = 'insert adjacency matrix';
2   n = size(D,1);
3   Z = zeros(n);
4   OUT = cumsum(D,2);
5   D1 = Z;
6   for i=1:n
7     D1 = D1 + logical(D^i);
8     D1=logical(D1);
9   end
10  D1(logical(eye(size(D1)))) = 0;
11  B = sum(transpose(D1));
12  B_all = sum(B);
13  DT = transpose(D);
14  DT1 = Z;
15  for i=1:n
16    DT1 = DT1 + logical(DT^i);
17    DT1=logical(DT1);
18  end
19  DT1(logical(eye(size(DT1)))) = 0;
20  C = sum(transpose(DT1));
21  C_all = sum(C);
22  for i=1:n
23    Ptwo(i,1) = B(i) - OUT(i,n);
24  end
25  Ptwo = max(sum(Ptwo),1);
26  for i=1:n
27    G = D;
28    G(:,i)=0;
29    G(i,:)=0;
30    GG = zeros(n);
31    for j=1:n
32     GG = GG + logical(G^j);
33     GG = logical(GG);
34    end
35   GG(logical(eye(size(GG)))) = 0;
36   B_2 = sum(transpose(GG));
37   B_2_all = sum(B_2);
38   NP(i,1) = B_all - B_2_all - (B(i) + C(i));
39  end
40  np = NP/Ptwo;
\end{verbatim}
\end{scriptsize}

\noindent
The second step is to characterise each node. From Theorem~\ref{undirectedmiddlemen} and Corollary~\ref{corundirectedmiddleman}, the distinction between weak and strong middlemen is easy to make. This allows us to enhance the algorithm for identifying the various classes of critical nodes in a network.

We first analyse the network power of the nodes in the directed network provided using the above algorithm. We then investigate the middleman power measure for the underlying undirected network. If a node has positive middleman power in the original network as well as a positive network power in the underlying network, then the node has to be a strong middleman. On the other hand, if the node has a positive middleman power in the original network, but zero middleman power in the underlying undirected network, that node has to be a weak middleman.

This is shown in the extended algorithm below.

\begin{scriptsize}
\begin{verbatim}
Language: MATLAB
Input: Adjacency matrix, D, and raw network power, NP.
Output: Middleman classification.

1  U = D + transpose(D);
2  U = logical(U);
3  U1 = Z;
4  for i=1:n
5      U1 = U1 + logical(U^i);
6      U1 = logical(U1);
7  end
8  U1(logical(eye(size(U1)))) = 0;
9  F = sum(transpose(U1));
10 F_all = sum(F);
11 for i=1:n
12    G = U;
13    G(:,i)=0;
14    G(i,:)=0;
15    GG = zeros(n);
16    for j=1:n
17        GG = GG + logical(G^j);
18        GG = logical(GG);
19    end
20    GG(logical(eye(size(GG)))) = 0;
21    F_2 = sum(transpose(GG));
22    F_2_all = sum(F_2);
23    NPU(i,1) = F_all - F_2_all - (F(i) + F(i));
24 end
25    for i=1:n
26    if NP(i) == 0 && NPU(i) == 0;
27        fprintf('Node %0.0f is a non-middleman\n',i);
28    elseif NP(i) > 0 && NPU(i) == 0;
29        fprintf('Node %0.0f is a weak middleman\n',i);
30    elseif NP(i) > 0 && NPU(i) > 0;
31        fprintf('Node %0.0f is a strong middleman\n',i);
32    end
33 end
\end{verbatim}
\end{scriptsize}

\noindent

\end{document}